\documentclass[lettersize,journal]{IEEEtran}

\usepackage{color}
\usepackage{soul}

\usepackage{balance}

\usepackage{flushend}

\usepackage{tabularx}
\usepackage{tabulary}
\usepackage{paralist}
\usepackage{caption}

\usepackage{lipsum,graphicx,subcaption}

\usepackage{pdflscape}

\usepackage{pifont}
\usepackage[table]{xcolor}
\usepackage{colortbl}
\usepackage{booktabs}
\usepackage{multirow}
\usepackage{mathtools}

\usepackage{amsmath,amsfonts,amsthm} 

\newtheorem{theorem}{Theorem}

\usepackage[linesnumbered,ruled]{algorithm2e}

\SetCommentSty{mycommfont}

\SetKwInput{KwInput}{Input}
\SetKwInput{KwOutput}{Output}
\SetKwFunction{FlowIndex}{FlowIndex}
\SetKwFunction{Flowrel}{Flowrel}
\SetKwFunction{Flowsend}{Flowsend}
\SetKwFunction{FlowConflict}{FlowConflict}

\SetKwRepeat{Do}{do}{while}

\usepackage{array}

\usepackage{lineno}
\usepackage{url}
\usepackage{filecontents}
\usepackage{hyperref}

\ifCLASSOPTIONcompsoc

  \usepackage[nocompress]{cite}
\else
  \usepackage{cite}
\fi

\begin{document}

\title{Multi-Source Coflow Scheduling in Collaborative Edge Computing with Multihop Network}

\author{Yuvraj Sahni, Jiannong Cao,~\IEEEmembership{Fellow,~IEEE,}, Lei Yang, and Shengwei Wang

\thanks{Yuvraj Sahni and Shengwei Wang are with the Department of Building Environment and Energy Engineering, The Hong Kong Polytechnic University, Hong Kong.  (E-mail: yuvraj-comp.sahni@polyu.edu.hk, beswwang@polyu.edu.hk)}
\thanks{Jiannong Cao is with the Department of Computing, The Hong Kong Polytechnic University, Hong Kong. (E-mail: jiannong.cao@polyu.edu.hk)}
\thanks{ Lei Yang is with the School of Software Engineering, South China University of Technology, Guangzhou, China. (E-mail: sely@scut.edu.cn)}
}

\maketitle

\begin{abstract}
Collaborative edge computing has become a popular paradigm where edge devices collaborate by sharing resources. Data dissemination is a fundamental problem in CEC to decide what data is transmitted from which device and how. Existing works on data dissemination have not focused on coflow scheduling in CEC, which involves deciding the order of flows within and across coflows at network links. Coflow implies a set of parallel flows with a shared objective. The existing works on coflow scheduling in data centers usually assume a non-blocking switch and do not consider congestion at different links in the multi-hop path in CEC, leading to increased coflow completion time (CCT). Furthermore, existing works do not consider multiple flow sources that cannot be ignored, as data can have duplicate copies at different edge devices. This work formulates the multi-source coflow scheduling problem in CEC, which includes jointly deciding the source and flow ordering for multiple coflows to minimize the sum of CCT. This problem is shown to be NP-hard and challenging as each flow can have multiple dependent conflicts at multiple links. We propose a source and coflow-aware search and adjust (SCASA) heuristic that first provides an initial solution considering the coflow characteristics. SCASA further improves the initial solution using the source search and adjust heuristic by leveraging the knowledge of both coflows and network congestion at links. Evaluation done using simulation experiments shows that SCASA leads to up to 83\% reduction in the sum of CCT compared to benchmarks without a joint solution.
\end{abstract}


\section{Introduction}

\IEEEPARstart{I}{n} the past decade, edge computing has become popular as a key technology to enable applications such as Smart Home, Smart City, Smart Healthcare, etc. Edge computing addresses the limitations of centralized cloud computing, such as high latency, limited scalability, privacy leakage, etc., by pushing computation, storage, and other services to edge devices closer to data sources. However, the proliferation of devices and data has led to the development of emerging applications such as autonomous vehicles, augmented and virtual reality, Industry IoT, etc., requiring collaboration among different system components. Therefore, researchers have recently proposed collaborative edge computing (CEC) that enables collaboration among edge devices by leveraging the sharing of different resources, including computation, data, and services \cite{dong2021collaborative} \cite{sahni2017edge}. 

One of the fundamental problems to support collaboration in CEC is optimizing the data dissemination among edge devices while sharing data and scheduling computation tasks. Data dissemination problem in CEC includes to deciding what data is transmitted from which device and how. Most existing work on data dissemination in context of edge computing have studied optimizing the decisions such as data access, data placement, routing, bandwidth allocation, etc., subject to system constraints \cite{aral2018decentralized} \cite{liu2020fog}, \cite{yang2019efficient}, \cite{singh2020intent} \cite{luo2020edgevcd}. However, few works have considered coflow scheduling, i.e., deciding the order of flows within and across coflows at congested network links, in CEC with a mesh network of edge devices connected using multi-hop paths. Coflow is the abstraction of parallel flows with shared objective, corresponding to the transmission of the multiple data requested at the edge devices \cite{chowdhury2014efficient}. Coflows can be observed in the case of scheduling both dependent and independent application tasks in edge computing. The different dependent tasks are often modeled as a directed acyclic graph (DAG), which implies a set of parallel flows corresponding to the transfer of data from preceding tasks to the current task. Even for independent tasks, each task can require multiple input data, leading to a set of parallel flows to transmit data. Such transmission of parallel flows is referred to as coflows. Instead of optimizing individual flows' completion time, the coflow completion time (CCT) optimizes the completion time of application tasks \cite{chowdhury2014efficient}. Furthermore, the requested data can be generated at different sources (edge devices), or the data can have duplicate copies at different sources that would require jointly considering the decision on selecting the source of different flows within each coflow. 

This paper studies the multi-source coflow scheduling problem in CEC to minimize the sum of CCT. The problem is novel as we consider making the joint decision on each flow's source and order among flows at different network links. Some existing works have studied the related problem of coflow scheduling in data centers, which also includes deciding the order among coflows to minimize the CCT \cite{chowdhury2014efficient} \cite{luo2016towards} \cite{ahmadi2020scheduling} \cite{chowdhury2019near}. However, the problem in this paper for CEC has some characteristic differences compared to existing coflow scheduling in data center networks. First, existing works usually assume a non-blocking switch that interconnects all physical hosts. However, each data flow in CEC is transmitted through a multi-hop network, leading to congestion in different links. Second, existing works consider flows can share bandwidth; however, for the problem in our paper, it is assumed that only one flow can occupy a link at one time. Third, existing works often assume fixed placement of flows, i.e., each flow is generated at a fixed source. However, there can be multiple sources for each data type in CEC, and the selection of sources can lead to different paths and conflicts. Furthermore, the data can be generated at different sources at different release times, leading to additional complexity. 

We have formulated the problem as an NP-hard optimization problem. The main challenge in the problem is to consider conflicts among multiple flows at different links to reduce the overall sum of CCT. Since there are multiple conflicting flows, each flow can suffer from multiple dependent conflicts at multiple different links in the path. Besides, the multiple conflicts at different links in the path can be among a different set of flows, leading to increased complexity. We have proposed a source and coflow-aware search and adjust (SCASA) heuristic to solve the problem. SCASA first decides the source and flow ordering by using an initial solution that considers global knowledge of coflows while ignoring network congestion at individual links. In the next step, we incorporate the knowledge of network congestion at links by proposing the source search and adjust heuristic that improves both the source selection and flow order to minimize the sum of CCT. SCASA has been extensively evaluated by comparing it with different benchmark solutions using simulation experiments. Performance comparison shows SCASA leads to significant performance improvement of up to 83\% reduction in the sum of CCT compared to different benchmark solutions. 

The coflow scheduling problem solved in this paper is practically useful for both sharing data across edge devices and scheduling tasks in different application domains. One application example is large-scale multi-camera video analytics, where video and image data from multiple cameras are transmitted at edge devices to generate situational awareness. Therefore, a set of parallel flows, i.e., coflow, are required to transmit multiple input data for the situational awareness task at different edge devices. Another important example is a building automation system for smart buildings, where executing the control operations for each subsystem (heating, ventilation, and air-conditioning (HVAC), lighting, security, energy, etc.) requires coflow to transmit input data from multiple collaborating subsystems. In both scenarios, there can be multiple redundant sensors for each input data type, implying multiple sources for each data flow. 

The main contributions of this work are: 

\begin{enumerate}
\item We mathematically formulate the problem for multi-source coflow scheduling in CEC to minimize the sum of CCT. The problem is formulated as a mixed-integer non-linear programming (MINLP) problem and shown to be NP-hard by reducing from a well-known NP-hard problem, a job-shop scheduling problem (JSSP) with the min-sum objective. To the best of our knowledge, we are the first to consider such a problem. 
\item We propose a heuristic algorithm SCASA that leverages the knowledge of coflows and network congestion at the links to decide source and flow order. SCASA first proposes an initial solution and then improves the source selection and flow ordering using the proposed source search and adjust heuristic. 
\item We do a comprehensive numerical evaluation of SCASA using simulation experiments in terms of the performance metric as the sum of CCT. SCASA has been compared with both random solution and relaxed solution obtained using the relaxation from MINLP to LP problem. We have also compared SCASA with several other benchmark solutions by varying different input parameters, including the number of sources, number of flows, number of coflows, the average release time of flows, and the number of devices in the network. 
\end{enumerate}

The rest of the paper is organized as follows: In Section 2, we provide the system model and problem formulation. In Section 3, we provide details on the proposed heuristic SCASA and its computation complexity. In Section 4, we explain the results obtained during performance evaluation. In Section 5, we discuss some related works. Finally, we give the conclusion in Section 6. 

\section{System Model and Problem Formulation}

This section describes the system model, including the network and coflow model, and the problem formulation. 

\subsection{System Model}

Fig \ref{f:edgemesh} shows the system architecture of Edge Mesh, an abstraction of CEC where the intelligence is distributed and pushed within the network by sharing computation resources and data between the mesh network of edge devices \cite{sahni2017edge}. An example scenario is illustrated in Fig \ref{f:edgemesh} where transmission of two coflows (marked red and blue) requested by two different devices leads to congestion at multiple links, shown by circles, in the network. Each coflow includes multiple flows, i.e., the transmission of multiple data types requested by the edge device marked as 1 and 2 in the example scenario. This paper considers that each flow can have multiple sources as the data can be generated at different redundant sources in CEC. The different redundant flow sources, i.e., the edge devices generating data, differ in terms of release time and transmission time. The data transmission corresponding to multiple flows leads to congestion and conflict among flows as the bandwidth is limited. 

The multi-source coflow scheduling problem in CEC studied in this paper involves minimizing the sum of CCT while deciding on the source and order of flows at each link. It is assumed that flows can be ordered at each link in the path as they traverse the multi-hop path from the source edge device to the destination edge device requesting data. The flow source selection and ordering decision is made at a centralized controller with global network knowledge, including flows and sources. The centralized controller, such as an SDN controller, is responsible for network monitoring, making coflow scheduling decisions, and sending it to each edge device that executes the decision. The system model does not consider the communication overhead of collecting the network information at the centralized controller. The system model assumes that edge devices are connected using a wireless network, however, we have solved the problem in an offline setting. We have not considered all the issues due to temporal dynamics in a wireless network and interference of wireless transmissions among neighboring devices. Nevertheless, these issues in the wireless network should be considered in future work. 

Some other assumptions have also been made in our model. First, the routing path for each flow is assumed to be known. Second, no two flows are allowed to pass through a link at the same time to consider the interference among simultaneous wireless transmissions. Third, the problem is formulated and solved for the static condition where the values of different parameters are known beforehand. The first two assumptions can be relaxed by including additional decision variables for routing path and bandwidth shared by different flows at each link. However, this would make the problem even more complex. The third assumption can be relaxed by solving the problem for online settings where future coflows are unknown. This paper focuses on multi-source, multi-coflow scheduling in a CEC environment, and these assumptions can be relaxed in future work. 

\begin{figure}[h]
\centering
\includegraphics[scale=0.35]{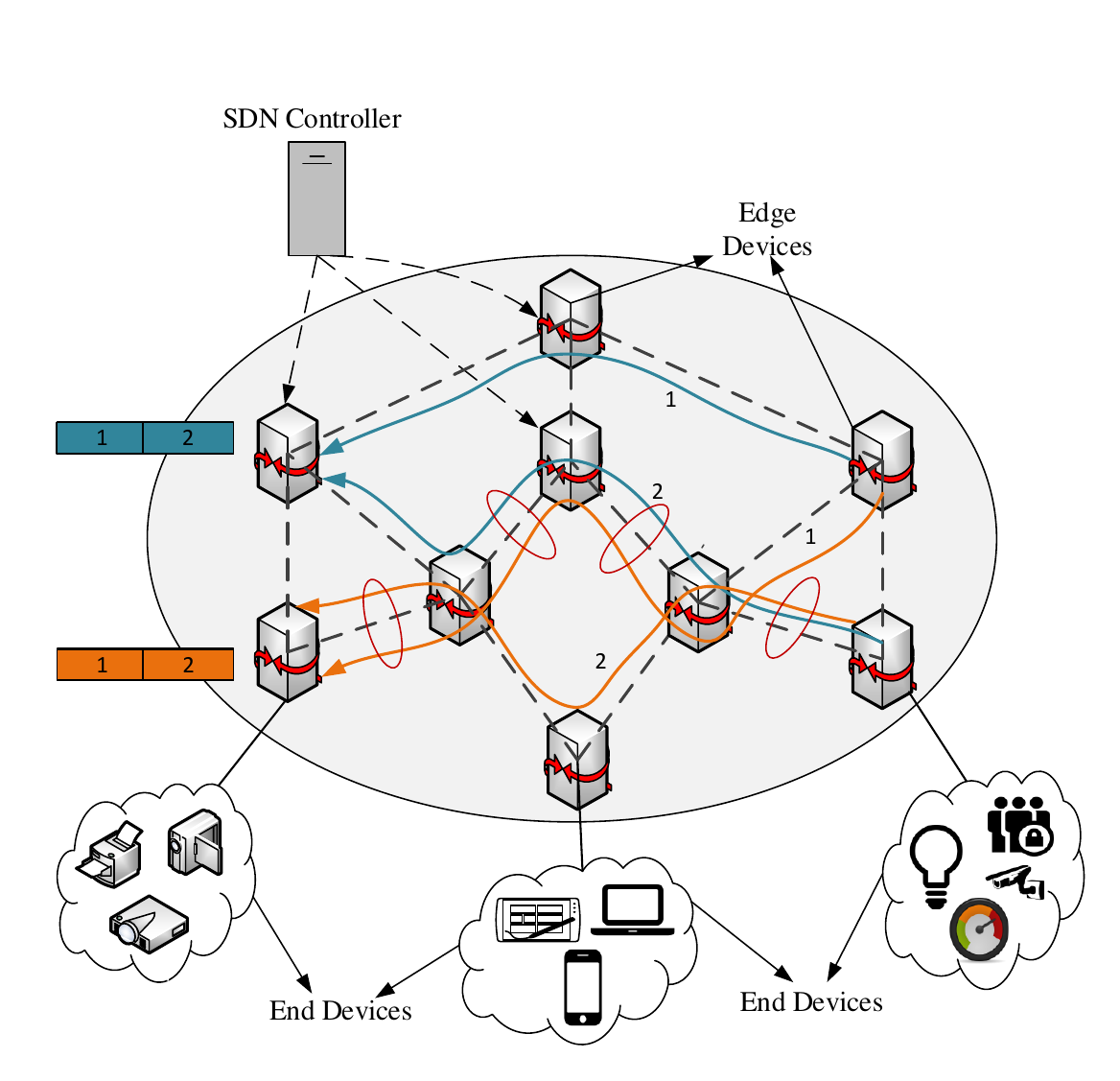}
\caption{System Architecture}
\label{f:edgemesh}
\end{figure}

Under the assumptions mentioned above, the network and coflow model used in formulating the problem are: 

\textit{Network model}: The communication network is a mesh network of edge devices connected using a multi-hop path. The communication network is modeled as a connected graph $G = (N1 \cup N2, E)$, where $N1$ is the set of devices requesting data, $N2$ is the set of the devices where different data is generated, and E is the set of links connecting different devices, $E = \{e_{iu} |i, u \in N1 \cup N2 \}$. Here, $e_{iu}$ corresponds to the bandwidth between devices $i$ and $u$. In this paper, we assume that a device requesting data will not generate data, i.e. $N1 \cap N2 = \emptyset$, to make it easier to formulate the problem. However, this assumption can be easily relaxed and will not affect the problem formulation and proposed solution. The total number of devices in the network is $N$, i.e. $|N1| + |N2| = N$.

\textit{Coflow model}: There are a set of coflows, represented by $N1$, in the network corresponding to edge devices requesting data. Each coflow $i$ in set $N1$ consists of a set of flows, $K_i$, as each edge device requests set $K_i$ of different data types. Each flow $j$ in set $K_i$ can be transmitted from an edge device $s$ in set $S_{ij}$ of sources, where $S_{ij} \subseteq N2$. The total number of flows within a coflow is $K$, while the number of sources for a flow is $S$. The different sources $s$ for a flow $j$ in coflow $i$ are assumed to have a different release time represented by $Trel_{ijs}$. Each flow $j$ in coflow $i$ from source $s$ is assumed to pass through a set of hops/links $H_{ijs}$. A binary parameter $L_{ijs}^{eh}$ (1 for yes, 0 for no) is used to represent if an edge $e$ in the network corresponds to hop $h$ for flow $j$ in coflow $i$ from source $s$. Each flow $j$ in coflow $i$ is associated with the amount of data to be transmitted, denoted by $D_{ij}$.  The bandwidth of the link for hop $h$ in the path of flow $j$ in coflow $i$ from source $s$ is denoted by $R_{ijs}^h$.  The cost of data transmission at hop $h$ for flow $j$ in coflow $i$ from source $s$ is denoted by $T_{ijs}^h$, as shown in Equation (1).

\begin{align}
T_{ijs}^h = \frac{D_{ij}}{R_{ijs}^h}
\end{align}

We do not include propagation time in communication cost as it is usually minimal. Further, we also ignore the switching cost for routing between subsequent links in the multi-hop path. In practice, these costs would influence the total cost; however, we ignore these costs to simplify the system model. Other works such as \cite{sundar2018offloading}, \cite{hong2019qos}, etc., have used similar assumptions to calculate the total cost. Therefore, the formulated problem and proposed solution in this paper can be directly extended by considering these costs together with transmission costs.

\subsection{Problem Formulation}

This section describes the constraints and formulates the problem as an optimization problem. There are two main decision variables used in formulating the problem: $X_{ijs}$ and $Y_{ijh}^{uvl}$. $X_{ijs}$ is a binary variable to denote the source selected for a flow. $X_{ijs} = 1$ means source $s$ is selected for flow $j$ in coflow $i$ and $X_{ijs} = 0$ otherwise. $Y_{ijh}^{uvl}$ is a binary variable to denote the order among two flows. $Y_{ijh}^{uvl} = 1$ means flow $j$ in coflow $i$ at hop $h$ precedes (or completes before) flow $v$ in coflow $u$ at hop $l$ and $Y_{ijh}^{uvl} = 0$ otherwise. $Ft_{ijh}$, the finish time of flow $j$ in coflow $i$ at hop $h$ is another decision variable that is derived based on the value of $X_{ijs}$ and $Y_{ijh}^{uvl}$.

\subsubsection{Constraints}

There are several constraints in the problem related to flow ordering across different hops in the path, the completion time of flow and coflows, and decision variables.  

\textit{Flow Ordering Constraints}: Equations (2)-(4) describe the constraints corresponding to the finish time of the flow at each hop. The problem considers that each flow can be ordered across different hops in the path. Equation (2) guarantees that for two different flows passing through the same network link, the finish time at the network hop of a flow is greater than the sum of data transmission at the network hop and the finish time at the network hop of the preceding flow.

\begin{align}
 & \begin{aligned}
& Ft_{ijh} + M*(2 - Y_{uvl}^{ijh} - X_{uvw}) >= T_{ijs}^h*X_{ijs} + \\
&\quad a_{ij}^{uv}*(Ft_{uvl}*L_{uvw}^{el}), \\
& \quad \quad \forall i, u\in N1, j \in K_i, v \in K_u, s \in S_{ij}, w \in S_{uv}, e \in P_{ijs}^h
\end{aligned}
\end{align}

The constraint in Equation (2) includes certain conditions: a) whether the two flows are the same or not, b) whether the other flow is preceding or not, and c) whether the two flows have a conflict (pass through the same link) otherwise the two flows can be transmitted simultaneously. Equation (2) uses a parameter $a_{ij}^{uv}$ to check whether the two flows are the same or not. $a_{ij}^{uv} = 1$ means flow $j$ in coflow $i$ is different from flow $v$ in coflow $u$and $a_{ij}^{uv} = 0$ otherwise. Another term in Equation (2) is $M*(2 - Y_{uvl}^{ijh} - X_{uvw})$ that determines whether the flow $u$ at hop $l$ in coflow $v$ precedes the flow $j$ in coflow $i$ at hop $h$. The third condition is checked using the parameter $L_{uvw}^{el}$ to decide whether link $e$ corresponds to transmission of flow $v$ in coflow $u$ from source $w$ at hop $l$. Here, link $e$ is the network link for transmission of flow $j$ at hop $h$ in coflow $i$ from source $s$. 

Each flow depends on the transmission of data across different hops in the path. For example, the transmission of flow at a hop can start only after the finish time of flow at the previous hop. Equation (3) shows this constraint by representing the relation between the finish time of flow at successive hops in the path. 

\begin{align}
& \begin{aligned}
& Ft_{ijh} \geq Ft_{ij(h-1)} + T_{ijs}^h*X_{ijs}, \\
& \quad \forall i \in N1, j \in K_i, s \in S_{ij}, h \in H_{ijs}
\end{aligned}
\end{align}

Equation (4) represents the lower bound on the finish time of flow at the first hop in the path as the sum of the data's release time and the cost of data transmission at the first hop. 

\begin{align}
& \begin{aligned}
&Ft_{ij1} \geq (Trel_{ijs} + T_{ijs}^1)*X_{ijs}, \quad \forall i \in N1, j \in K_i, s \in S_{ij}
\end{aligned}
\end{align}

\textit{ Coflow Constraints}: Equations (5) and (6) represent the constraints in the problem formulation corresponding to the completion time of flows and coflows. Equation (5) guarantees that flow completion time (FCT) is greater than the finish time of flow at any hop in the path. Similarly, the coflow completion time (CCT) is greater than the completion time of all flows within the coflow.

\begin{align}
& \begin{aligned}
&FCT_{ij} \geq Ft_{ijh}, \quad \forall i \in N1, j \in K_i, h \in H_{ijs}, s \in S_{ijs}
\end{aligned}\\
& \begin{aligned}
& CCT_i  \geq FCT_{ij},  \quad \forall i \in N1, j \in K_i
\end{aligned}
\end{align}

\textit{Decision Variable Constraints}: Equation (7)-(9) represents the constraints in the problem formulation corresponding to decision variables for selecting the source and order of flows. Equation (7) guarantees that exactly one source is selected for each flow. Equation (8) constraints variable values for flow order such that one flow has to precede the other at any hop. For example, if $Y_{ijh}^{uvl} = 0$ implying flow $j$ in coflow $i$ at hop $h$, does not precede the flow $v$ in coflow $u$ at hop $l$, then the latter flow should precede the former flow, i.e. $Y_{uvl}^{ijh} = 1$, or vice versa. Equation (9) represents that the range of decision variables, $X_{ijs}$ and $Y_{ijh}^{uvl}$, is binary.  

\begin{align}
& \begin{aligned}
& \sum_{s = 1}^{S} X_{ijs} = 1,  \quad \forall i \in N1,  j \in K_i
\end{aligned}\\
& \begin{aligned}
&Y_{ijh}^{uvl} + Y_{uvl}^{ijh} = 1,  \quad \forall i, u\in N1, j \in K_i, v \in K_u, \\ 
& \quad h \in H_{ijs},  s \in S_{ij}, l \in H_{uvw}, w \in S_{uv} \\
\end{aligned}\\
& \begin{aligned}
&X_{ijs}, Y_{ijh}^{uvl} = \{0, 1\},  \quad \forall i, u\in N1, j \in K_i, v \in K_u, \\ 
& \quad s \in S_{ij}, h \in H_{ijs}, l \in H_{uvw}, w \in S_{uv} \\
\end{aligned}
\end{align}

 \subsubsection{Optimization Problem}

The problem is defined as follows: Given the network model and coflow model, the objective of the problem is to minimize the sum of CCT for all edge devices requesting data by making the decision on the source and order of flows at each hop in the path. The multi-source coflow scheduling problem can be formulated as follows:

\begin{align}
& \underset{X_{ijs},Y_{ijh}^{uvl}, Ft_{ijh} }{\text{Minimize}} \sum_{i = 1}^{|N1|} CCT_{i}, \\
& \text{Subject to}: \quad \text{Equations } (2) - (9) \nonumber
\end{align}

In the objective function, $CCT_i$ expresses the CCT for flows corresponding to data requested by edge device $i$. The objective function is to minimize the sum of CCT for all edge devices requesting data. The problem is formulated as a mixed integer non-linear programming (NLP) problem, where $X_{ijs}$ and $Y_{ijh}^{uvl}$ are discrete decision variables with binary values, while $Ft_{ijh}$ is continuous decision variable with positive real value. 

\begin{theorem}
Problem \textbf{P1} is NP-hard
\end{theorem}

\begin{proof}
We prove the formulated problem is NP-hard by reducing it from a well-known NP-hard problem, a job-shop scheduling problem (JSSP) with min-sum objective \cite{queyranne2002approximation} to a special instance of the current problem. 

We consider the special instance of the current problem where each flow has a single source, and each coflow consists of a single flow. Each job in JSSP corresponds to coflow in the special instance of the current problem, the operation in a job represents the transmission of flow across each hop, and the machine in JSSP is equivalent to the link used for transmitting the data. The processing time of operation in a job at the machine is equivalent to transmission time for flow at each hop. The processing time of operations is heterogeneous as each flow has a different amount of data to be transmitted. The different operations, i.e., data transmission across consecutive hops for a flow, have a precedence order. Each operation needs to be processed on a specific machine, and only one operation in a job can be processed at a given time. The objective is to determine the order of operations of jobs at machines (flow order at different hops) to minimize the sum of completion time of jobs (sum of CCT). Since JSSP with min-sum objective has been shown to be NP-hard \cite{queyranne2002approximation}, our problem is also NP-hard. 
\end{proof}

 \section{Source and Coflow-Aware Search and Adjust (SCASA) Heuristic} \label{s:scasa}
 
This section gives details of the SCASA heuristic to solve the multi-source coflow scheduling problem in CEC. SCASA gives an offline schedule for coflows that includes selecting each flow's source and deciding the order among conflicted flows at different hops/links in the network path to minimize the sum of CCT.  The main idea of SCASA is to find an initial solution for selected source and flow ordering at different links and then improve it using a source search and adjust heuristic. The source search and adjust heuristic modifies the source of flows and order among flows to minimize the sum of CCT. The details of each part of SCASA are explained in the following subsections.

\subsection{Initial Solution}

The initial solution includes first determining the source of all flows based on the rank metric and then deciding the priority order of flows at different hops in the network path. Each flow $j$ in coflow $i$ is given a rank, $Frank_{ij}$ as shown in equation (11). The source of each flow is decided independently, considering both the release time and the sum of transmission time at each hop in the path. The condition for deciding the source of each flow is to select one that leads to the minimum value of $Frank_{ij}$ metric, i.e, $\min_{s \in S_{ij}}{Trel_{ijs} + \sum_{h = 1}^{|H_{ijs}|}T_{ijs}^h}$. 

\begin{align}
& \begin{aligned}
&Frank_{ij} =  \min_{s \in S_{ij}}{Trel_{ijs} + \sum_{h = 1}^{|H_{ijs}|}T_{ijs}^h} \\
\end{aligned}\\
& \begin{aligned}
&Crank_{i} =  \max_{j \in K_i}{Frank_{ij}} \\
\end{aligned}
\end{align}

The order among flows is decided based on the knowledge of coflows. The main idea behind deciding the priority order is to consider the rank of coflows instead of independent flows. Here, the rank of each coflow $i$, $Crank_i$ shown in equation (12), is determined based on the maximum transmission time among all flows, where each flow is transmitted without considering network congestion. The different flows are then prioritized in ascending order of corresponding coflow rank, i.e., flows belonging to the higher priority coflow are given higher priority. Since different flows corresponding to the same coflow can also pass through the same network link, we also need to determine their priority order. The priority order among flows within the same coflow is determined based on ascending order of $Frank$.

Once the source and flow order at different hops is determined, the CCT and sum of CCT can be calculated using Algorithm \ref{a:cct}. Algorithm \ref{a:cct} starts by creating a directed graph, $G1 = (V1,E1)$, to represent flow ordering. Each vertex in $V1$ corresponds to flow passing through a hop a network path, referred to hereafter as subflow. An edge between two vertices represents the priority ordering of subflows, i.e., if one subflow needs to finish before the other. The weight of the edge is equal to the transmission time of the subflow starting earlier. The different vertices are connected using the flow ordering determined earlier. There is also an additional vertex to represent the release of all flow. This additional vertex is connected to the first subflow (first hop in the network path of flow) in each flow, where the weight of an edge is equal to the release time of flow. Each vertex is associated with a start and finish time that must be calculated. The start time and finish time for additional vertex corresponding to the release of flows are zero. The dependency ordering among different subflows is found using topological sorting of $G1$. The start time ($Sthop_h$) and finish time ($Fthop_h$) corresponding to each subflow vertex ($h$) are determined as shown in lines 3-9. The finish time for each flow, CCT, and the sum of CCT can be calculated accordingly, as shown in lines 10-18. 

\begin{algorithm}
\footnotesize
\caption{Calculate CCT}
\label{a:cct}
\KwInput{Selected source of each flow and precedence order among the flows at different hops}
\KwOutput{Finish time of each flow, CCT, and sum of CCT}
 \BlankLine
 Create directed graph, $G1$, to represent flow ordering\;
Get dependency order of subflows, $O$, using topological sort\;
 \For{$a \gets 2 \text{ to } |V1|$ }{
$ h \gets O(a)$\;
 \text{Get index of coflow $i$ and flow $j$ for subflow h}\;
$pd_{h} \gets \text{set of predecessors of $h$ in $G1$}$\;
$Sthop_{h} \gets \max_{p \in pd_{h}}(Sthop_{p} + T_{ijs}^h)$\;
$Fthop_{h} \gets Sthop_h + T_{ijs}^h$
 }
   \For{$i \gets 1 \text{ to } |N1|$ }{
   \For{$j \gets 1 \text{ to } |K_i|$ }{
  $Ft_{ij} \gets \max_{h \in H_{ijs}}(Fthop_h)$\;
  }
  }
  \For{$i \gets 1 \text{ to } |N1|$ }{
  $CCT_i \gets \max_{j \in K_i}(Ft_{ij})$\;
  }
  $SumCCT \gets sum(CCT_i)$\;
 \KwRet $Ft_{ij}$, $CCT_i$, $SumCCT$
\end{algorithm}

\subsection{Source Search and Adjust}

Since the initial solution selects each flow's source independently without considering network congestion due to other flows passing through the same network link, it can lead to inefficient performance. Therefore, the proposed source adjustment greedily changes each flow's source to minimize the sum of CCT while keeping the source of other flows unchanged. One main issue here is deciding the priority flow order for changing sources. Our solution orders the flows according to a metric such that flows contributing more towards network congestion are given higher priority for changing the initially selected source. Here, the metric used to quantify the network congestion for each flow is the number of links in the flow path. The intuition behind the priority ordering metric is that a flow passing through more links will lead to network congestion, i.e., delay in transmission time, for more flows. Hence, changing the source of a higher priority flow can decrease network congestion for more flows, consequently minimizing the sum of CCT.

\begin{algorithm}
\footnotesize
\caption{Source Search and Adjust}
 \label{a:source}
\KwInput{Initial Schedule with selected sources, order of flows, CCT, and the sum of CCT}
\KwOutput{New Schedule with selected sources, order of flows, CCT, and sum of CCT}
 \BlankLine
Calculate the cost matrix of flows, $flcost$, where each row corresponds to the number of subflows in each flow, \;
Calculate the conflict matrix of flows, $flconf$, where each row corresponds to the number of conflicts with other flows\;
 \For{$i1 \gets 1 \text{ to } |N1*K|$ }{
$flidx \gets \max(flcost)$\; 
   \uIf{$flconf(flidx) \geq 1$}{ 
 Get index of coflow $i$, flow $j$, and source $s$ corresponding to $flidx$ \;
 $sumcct$ $\gets$ current sum of CCT using source $s$\;
  \For{$v \gets 1 \text{ to } S$ }{
  \uIf{$v == s$}{ 
  $sumcct_v \gets sumcct$\;
  }
 \Else{
Get new rank of coflow $i$ and flow $j$\;
Get new flow ordering based on new ranks using the same method as the initial solution\; 
Calculate $sumcct_v$ for new source $v$ using Algorithm \ref{a:cct}\;
}
  }
   \uIf{$sumcct == \min(sumcct_v)$}{ 
  Keep the source as initial $s$\;
   }
   \Else{
Set new source according to $\min(sumcct_v)$\;
  }
  Set the values of $sumcct$, rank of flows and coflows according to selected source\;
  Set row of $flidx$ in $flcost$ to zero\;
  Recalculate conflict matrix $flconf$\;
  }
  }
 \KwRet $X_{ijs}$, $\text{flow ordering graph}$, $Ft_{ij}$, $CCT_i$, $SumCCT$
\end{algorithm}

Algorithm \ref{a:source} shows the pseudocode of the proposed source search and adjust procedure. Given the initial schedule as input, Algorithm \ref{a:source} outputs a new schedule with newly selected sources, order of flows, CCT, and the sum of CCT. The algorithm starts by creating a cost matrix of flows $flcost$ and a conflict matrix of flows $flconf$ in lines 1 and 2. From lines 3 to 25, each flow is considered for changing sources based on the highest value in $flcost$ that corresponds to the flow with the maximum number of links in the path among the current flows not yet considered for changing sources. The selected flow, $flidx$, is further checked for conflict with other flows within the loop. If there is no conflict, i.e., no other flow passes through the links in the path of $flidx$, then there is no need to change the flow source as the initial solution already selected the source with minimum transmission time. For flow with network conflicts, the sum of CCT is calculated for different sources, i.e.. $sumcct_v$, using Algorithm \ref{a:cct}. The flow ordering input to calculate $sumcct_v$ using Algorithm \ref{a:cct} is based on the coflow and flow rank metrics method, similar to the initial solution. The source is kept unchanged if the minimum value of $sumcct_v$ is the same as the initially selected source. Otherwise, the source with a minimum value of $sumcct_v$ is selected for $flidx$. The values of $sumcct$, the rank of flows, and coflows are set accordingly. The cost matrix, $flcost$, and conflict matrix, $flconf$, are also updated. The process of changing the source of flows is iterated until all flows have been considered.

\subsection{Complexity Analysis}

There are two main parts in calculating the computation complexity of SCASA, i.e., the initial solution and the source search and adjust heuristic. 

First, the computation complexity of the initial solution is $\mathcal{O}(N1^2*K^2*H)$, which is determined based on the Calculate CCT shown in Algorithm \ref{a:cct}. There is one for loop in Algorithm \ref{a:cct} corresponding to the number of vertices in graph $G1$, which is $N1*K*H$ in the worst case. Here $H$ is the maximum number of hops in the routing path of the given network model $G$. Furthermore, another loop corresponds to the number of predecessors in line 7, i.e., $N1*K$ in the worst case. 

Second, the computation complexity of the source search and adjust heuristic shown in Algorithm \ref{a:source} is $\mathcal{O}(N1*K*S*(N1^2*K^2*H))$. There is one for loops of $N1*K$ corresponding to the number of flows and another loop of $S$ for the number of sources. The calculation of the sum of CCT in line 14 contributes most towards computation complexity, i.e., $N1^2*K^2*H$. 

Since the source adjustment heuristic contributes most towards computation complexity, the overall computation complexity of SCASA is $\mathcal{O}(N1*K*S*(N1^2*K^2*H))$.

 \section{Performance Evaluation}
 
We have done simulation experiments to evaluate and compare the performance of SCASA with other benchmark solutions. The evaluation has been done in terms of the sum of CCT averaged over different iterations as the performance metric. The parameters used for simulation are in a similar range to the one used previously in \cite{sundar2018offloading} and \cite{yang2013task}. We have conducted the simulation experiments on a Macbook laptop with 32 GB RAM and an Apple M2 pro processor with 10 cores.

\subsection{Simulation Parameters} \label{s:param}

\textit{Parameters for Network Model}: We generate a network of edge devices where devices are deployed randomly using a uniform distribution. The size of the area is selected to be $N\times N$ \text{square units}, and any two devices less than $2*N/10$ \text {units} apart are connected. The distance between devices is set to be in a similar range as done in previous works \cite{yang2013task}, \cite{sahni2018data}. A variable area size makes creating a connected mesh network topology easier even with a low number of devices\cite{sahni2020multi}. The devices are connected using a multi-hop path to form a connected graph. The bandwidth of each link is selected from a normal distribution with a mean of 20Mbps and a variance of 30\%.

\textit{Parameters for Coflow Model}: The total number of devices in the network is set as 40 in the default case, with the number of devices requesting data, i.e., coflows, set as 20 and the remaining devices set for generating data. The number of flows in a coflow is set to be 3 in the default case. There are 3 source alternatives for each flow in the default case. The amount of data is selected from a normal distribution with a mean of 2Mb and variance of 30\%. The release time of data is selected from a normal distribution with a mean $T_{avg}*\frac{D_{avg}}{B_{avg}}$ and a variance of 30\%. Here, $D_{avg}$ is the mean amount of data, $B_{avg}$ is the mean bandwidth of the link, and $T_{avg}$ is a scaling parameter set to 1 in the default case.

 \subsection{Benchmark Solutions} \label{s:benchmark}
 
 We have compared the proposed solution SCASA with six different benchmark solutions:
 
\begin{enumerate}
\item RANDOM: RANDOM solution implies that the flow source selection and order are determined randomly. 
The computation complexity of RANDOM is based on calculating the sum of CCT, i.e., $\mathcal{O}(N1^2*K^2*H)$. 
\item RELAXED: RELAXED solution involves solving the relaxed LP problem using a Mosek solver in CVX and then discretizing the fractional decision variables. 
\item Flow Scheduling (FLS): FLS includes independently selecting each flow's source using the method in the initial solution in SCASA. The order of flows is determined based on the priority metric in ascending order of the sum of release time and transmission cost using the selected source. The computation complexity of FLS is based on calculating the sum of CCT, i.e., $\mathcal{O}(N1^2*K^2*H)$. 
\item Coflow Scheduling (CFLS): CFLS uses the method in the initial solution in SCASA to determine both the source and flow order. Compared to FLS, CFLS uses knowledge of other flows within the coflow to determine the order, similar to the smallest-effective-bottleneck-first (SEBF) approach proposed in \cite{chowdhury2014efficient}. The computation complexity of CFLS is based on calculating the sum of CCT, i.e., $\mathcal{O}(N1^2*K^2*H)$. 
\item Bottleneck aware scheduling (BAS): BAS uses the same method in the initial solution in SCASA to determine the source of flows. The order among flows is determined based on the 2-approximation solution proposed in \cite{luo2016towards}. The computation complexity of BAS is based on calculating the sum of CCT, i.e., $\mathcal{O}(N1^2*K^2*H)$. 
\item Flow ordering-based heuristic (FLORD): FLORD selects the source of flows using the method in the initial solution in SCASA and then uses the flow adjustment that searches and iterates across network links to modify the flow ordering at network links, similar to shifting bottleneck heuristic \cite{adams1988shifting}. The computation complexity of FLORD is based on the flow adjustment heuristic in SCASA, i.e., $\mathcal{O}(E^2*(N1^2*K^2*H))$.

\end{enumerate}

\subsection{Comparison in Default Parameter Setting}

\begin{table*}
\caption{Performance Comparison in Default Parameter Setting}\label{t:default}
\begin{center}
 \begin{tabulary}{1.05\textwidth}{ @{} LLLLLLLLL @{}m{0pt}@{}}
  \toprule
  \textbf{Metric} & \textbf{RANDOM} &  \textbf{RELAXED} & \textbf{FLS}  & \textbf{CLFS} &  \textbf{BAS} &  \textbf{FLORD} &  \textbf{SCASA}\\
 \midrule
Sum of CCT (seconds) & 354.5 $\pm$ 155.1 & 148.7 $\pm$ 81.2 & 72.9 $\pm$ 30.4 & 65.7 $\pm$ 24.2 & 69.2 $\pm$  24.1 & 81.0 $\pm$  37.9 & 57.8 $\pm$ 19.7\\
Running time (seconds) & 0.2 $\pm$ 0.08  & 1.11E+05 $\pm$ 1.4E+04 & 0.09 $\pm$ 0.03& 0.08 $\pm$ 0.02 & 0.07 $\pm$ 0.02 & 26.3  $\pm$ 5.13& 9.02 $\pm$ 3.10 \\
 \bottomrule
  \end{tabulary}
\end{center}
\end{table*}

Table \ref{t:default} shows the performance comparison for default parameters, where the number of edge devices in the network is 40, the number of coflows is 20, the number of flows within each coflow is 3, and the number of sources for each flow is 3. We have averaged the results by conducting the experiment for 30 iterations, and the error margin is calculated for a 95\% confidence interval. It can be observed that SCASA leads to the best performance in terms of reducing the value of the sum of CCT. SCASA leads to a reduction in the value of the sum of CCT by 83.8\%, 61.4\%, 21.3\%, 12.6\%, 17.1\%, and 29.1\% compared to RANDOM, RELAXED, FLS, CFLS, BAS, and FLORD respectively. 

RANDOM performs the worst of all the benchmarks as it randomly selects the source and flow order. RELAXED performs better than RANDOM; however, it also loses a lot of information in discretizing the obtained fractional output of variables, leading to worse performance. The four other benchmarks, FLS, CFLS, BAS, and FLORD, use different approaches to find the source and flow ordering without considering the adjustment of the source, leading to worse performance than SCASA. It shows that the proposed source search and adjust heuristic in SCASA contributes towards performance improvement. 

While SCASA leads to better performance in terms of reducing the value of the sum of CCT, it has a much higher running time compared to benchmarks, FLS, CFLS, and BAS, which do not consider any adjustment after the initial solution. However, SCASA has a lower running time than RELAXED and FLORD. There is a tradeoff between the two metrics, the sum of CCT and running time, for different algorithms. It should be noted that the running time of algorithms depends on the available processing resources and algorithm implementation. Since the focus of the evaluation was to verify the efficacy of the proposed solution, an efficient implementation of the solution would be able to significantly reduce the running time.

\subsection{Comparison with Random and Relaxed solution}
Fig \ref{f:randrelax} shows the performance comparison of SCASA with two benchmark solutions, RELAXED and RANDOM, by changing the number of devices in the network and coflows. We have fixed the ratio of the number of devices in the network to coflows as 2:1. The number of devices increased from 10 to 40, while the number of coflows increased from 5 to 20. The number of flows within each coflow is set to 3, and the number of sources for each flow is set to 3. The value of the sum of CCT has been averaged over 30 iterations. 

\begin{figure}[h]
\centering
\includegraphics[scale=0.35]{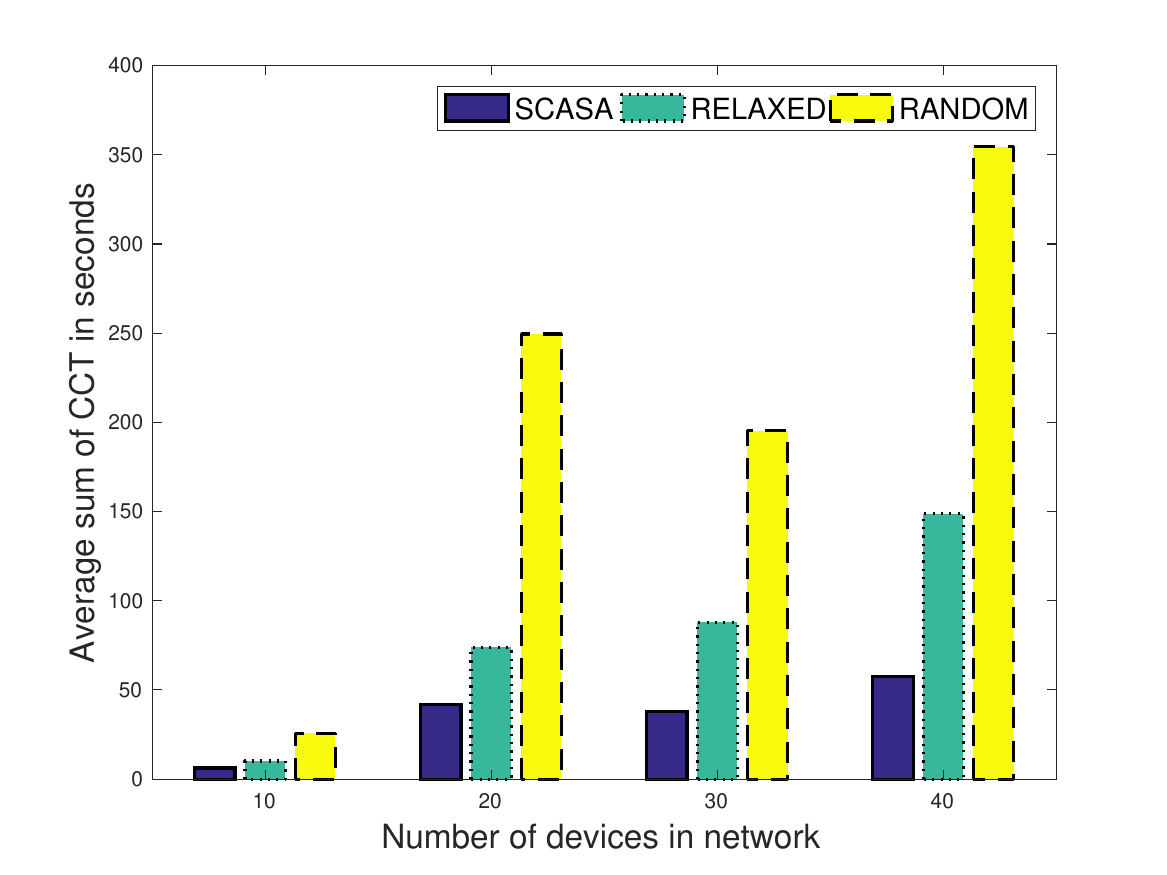}
\caption{Comparison by changing the number of devices in network and coflows}
\label{f:randrelax}
\end{figure}

It can be observed from Fig \ref{f:randrelax} that SCASA performs significantly better than both RELAXED and RANDOM solutions. The value of the average sum of CCT by SCASA is at least 41.5\% and 75.3\% less compared to RELAXED and RANDOM, respectively. Furthermore, the performance difference increases on increasing the number of devices and coflows. The performance difference between SCASA and RELAXED increases from 41.5\% to 61.4\% on increasing the number of devices from 10 to 40. Similarly, there is an increase in the performance difference between SCASA and RANDOM from 75.3\% to 83.4\% on increasing the number of devices from 10 to 40.

\subsection{Comprehensive Comparison with Benchmark Solutions}

This subsection shows the performance comparison of SCASA against four different benchmark solutions, i.e., FLS, CFLS, BAS, and FLORD. RANDOM and RELAXED are not compared as their performance is significantly worse than SCASA, as shown in Fig  \ref{f:randrelax}. Furthermore, RELAXED has a very high running time compared to SCASA. Unless stated otherwise, the default values used for performance comparison are: number of devices in the network is set to 40, number of coflows is set to 20, number of flows in each coflow is set to 3, number of sources for each flow is set to 3, and scaling parameter $T_{avg}$ of release time is set to 1. The value of the sum of CCT is averaged over 30 iterations.

 \begin{figure*}[!h]
 \begin{subfigure}[t]{0.33\textwidth}
    \centering
    \includegraphics[width=1\textwidth]{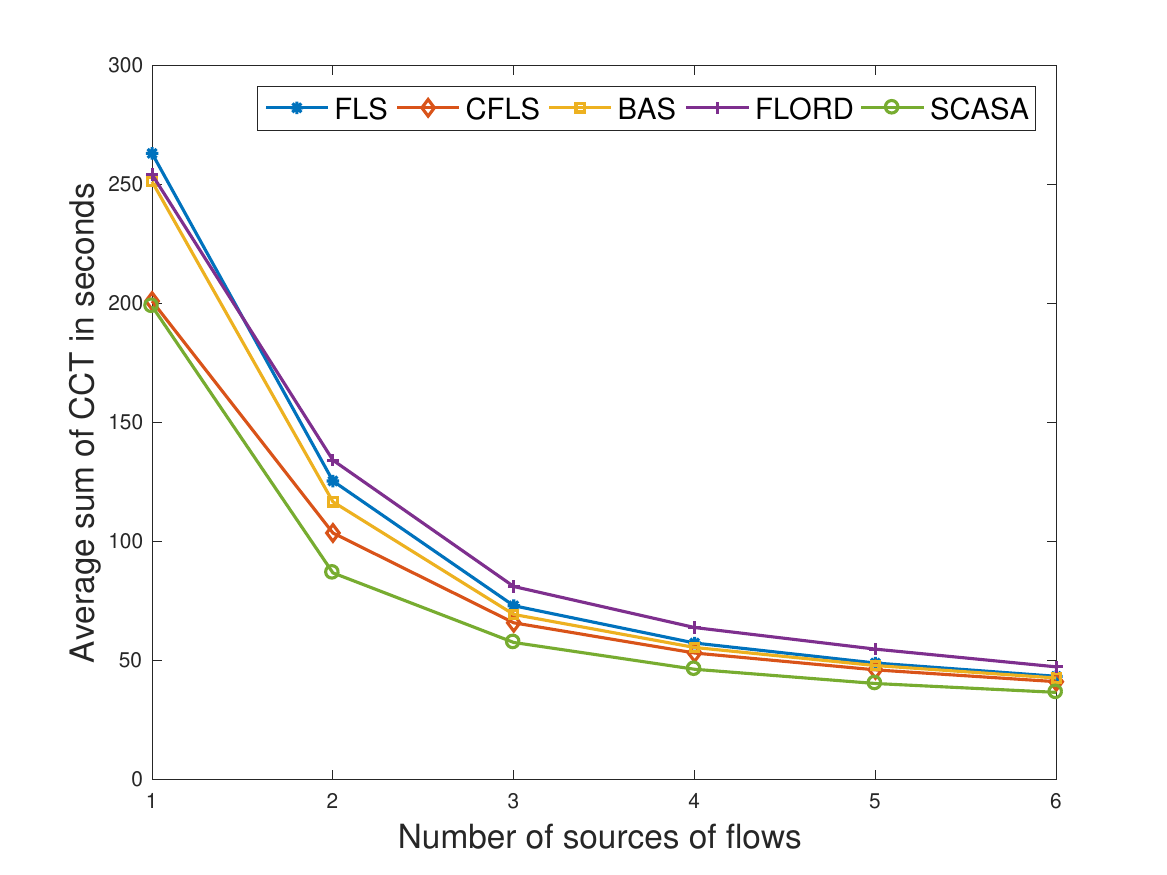}
    \caption{Changing number of sources}
    \label{f:sources}
\end{subfigure}  
~
 \begin{subfigure}[t]{0.33\textwidth}
    \centering
    \includegraphics[width=1\textwidth]{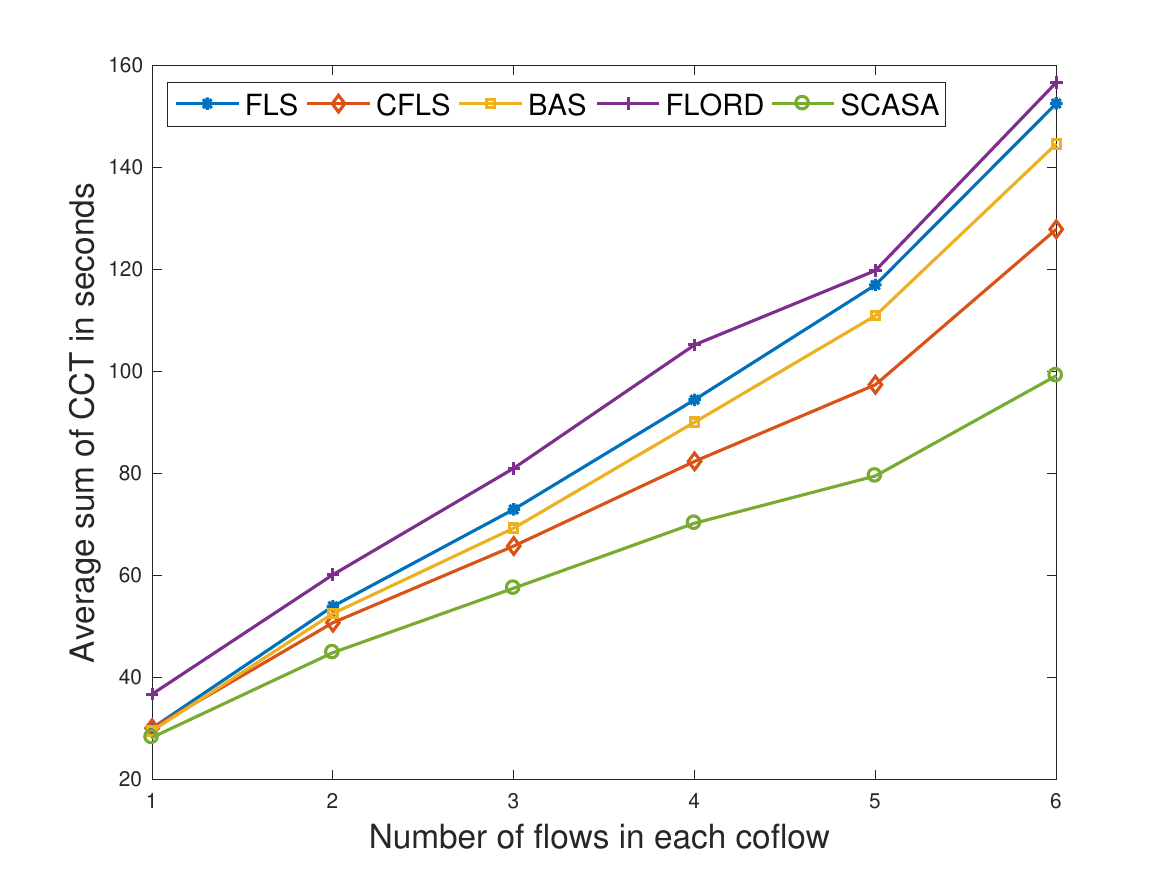}
    \caption{Changing number of flows}
    \label{f:flows}
\end{subfigure}  
~
 \begin{subfigure}[t]{0.33\textwidth}
    \centering
    \includegraphics[width=1\textwidth]{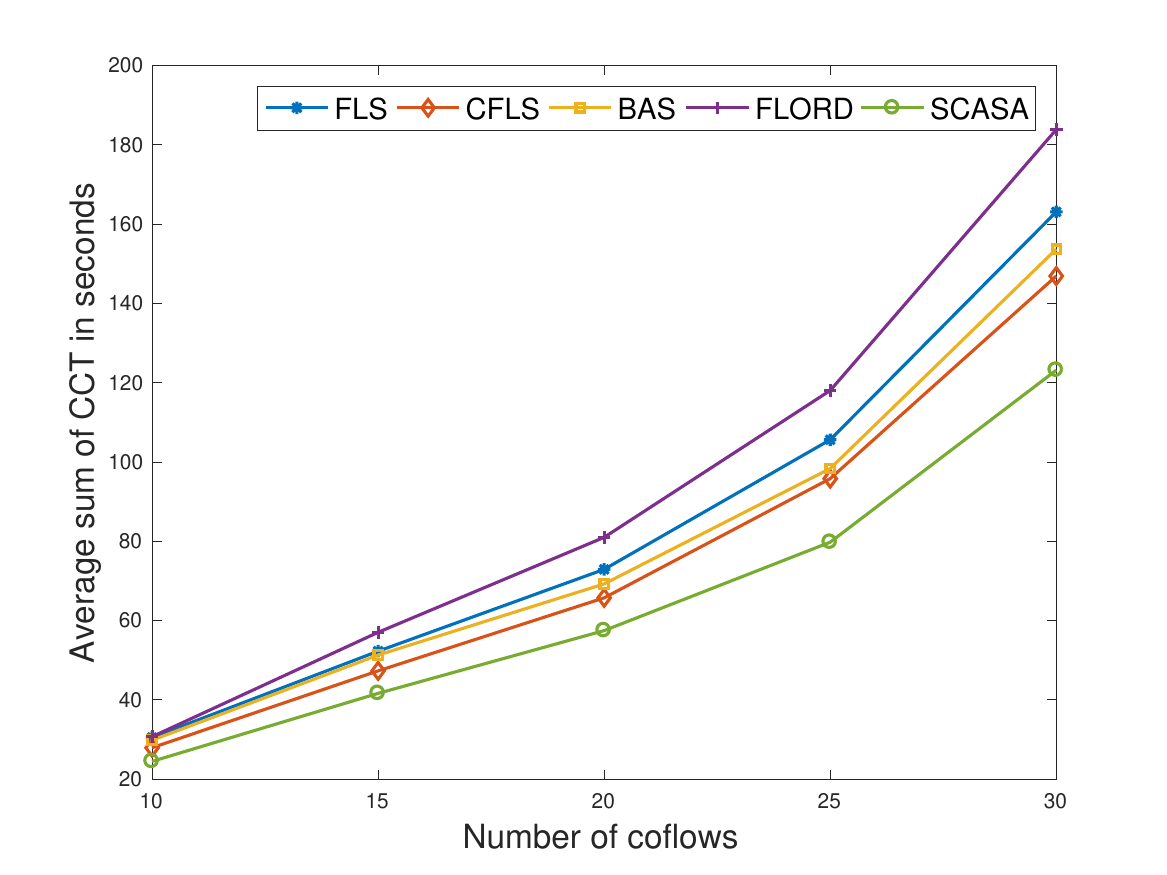}
    \caption{Changing number of coflows}
    \label{f:coflows}
\end{subfigure}

 \begin{subfigure}[t]{0.33\textwidth}
    \centering
    \includegraphics[width=1\textwidth]{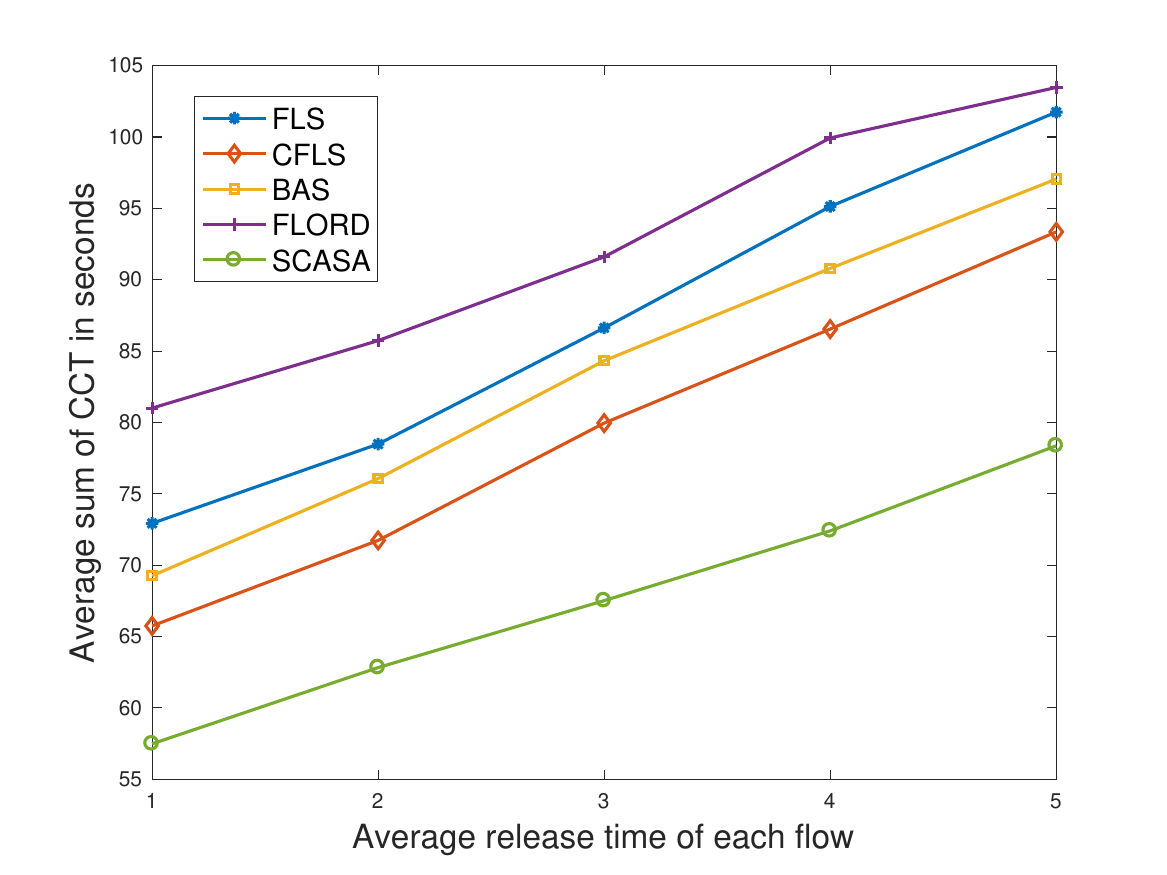}
    \caption{Changing average release time of data}
    \label{f:release}
\end{subfigure}   
~
 \begin{subfigure}[t]{0.33\textwidth}
    \centering
    \includegraphics[width=1\textwidth]{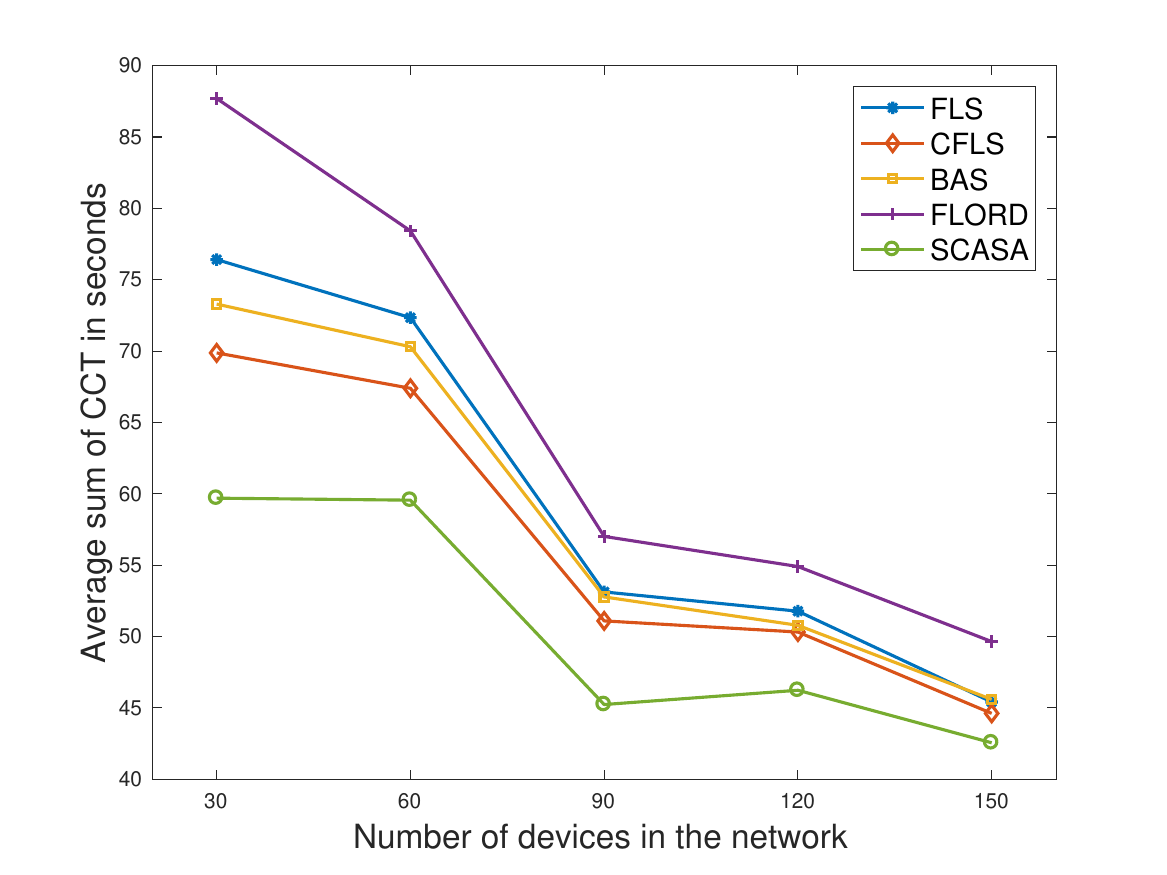}
    \caption{Changing number of devices in network}
    \label{f:devices}
\end{subfigure}   
~  
 \begin{subfigure}[t]{0.33\textwidth}
    \centering
    \includegraphics[width=1\textwidth]{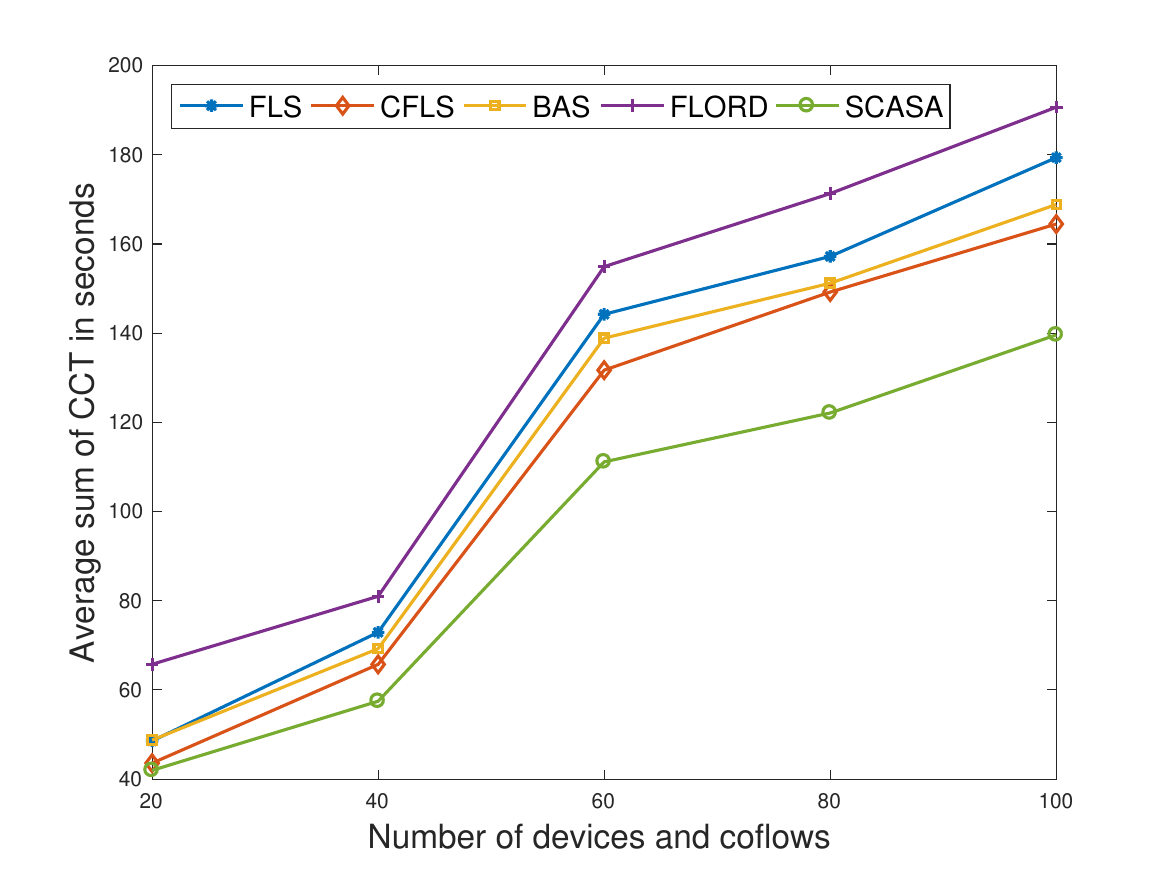}
    \caption{Changing both number of devices and coflows with ratio as 2:1}
    \label{f:devicesreqnet}
\end{subfigure}   
 \caption{Performance comparison with benchmark solutions}
\label{f:devcof}
\end{figure*}

\subsubsection{Changing number of sources}

Fig \ref{f:sources} shows the performance comparison between SCASA and four benchmark solutions by changing the number of sources for each flow from 1 to 6. It can be observed that SCASA leads to better performance, i.e., a lower value of the average sum of CCT, compared to all benchmark solutions for all values of the number of sources. The performance difference between SCASA and other benchmark solutions is up to 30.8\%, 16.2\%, 25.6\%, and 35.3\% compared to FLS, CFLS, BAS, and FLORD, respectively. SCASA performs better due to a better initial solution and the proposed source search and adjust heuristic that considers the network conflicts. The performance difference is decreased to 15.60\%, 11\%, 14.20\%, and 22.70\% compared to FLS, CFLS, BAS, and FLORD, respectively, on increasing the number of sources to 6. The decrease in performance difference when increasing the number of sources is due to better source selection, leading to the reduced value of the sum of CCT. As the number of sources increases, there are fewer conflicts at the network links as sources closer to devices requesting data are selected. 

\subsubsection{Changing number of flows}

Fig \ref{f:flows} shows the performance comparison between SCASA and benchmark solutions on changing the number of flows within each coflow from 1 to 6. As the number of flows increases, the network congestion at the links also increases. Therefore, the increase in the number of flows increases the performance difference between SCASA and other benchmark solutions. The performance difference between SCASA and FLS increases from 6\% at 1 flow to 34.9\% at 6 flows. For CFLS, the performance difference increases from 6\% to 22.5\% on increasing from 1 to 6 flows within each coflow. Similarly, there is an increase in performance difference from 3.6\% to 31.4\% for BAS and from 23.1\% to 36.7\% for FLORD. Among different benchmark solutions, both BAS and CFLS perform better as they consider the knowledge of coflows to determine flow order. However, both CFLS and BAS perform worse than the proposed SCASA as it further improves the performance by using source search and adjust heuristic to adjust the source selection and flow ordering at network links.

\subsubsection{Changing number of coflows}

Fig \ref{f:coflows} shows the performance comparison between SCASA and benchmark solutions on increasing the number of coflows from 10 to 30. Similar to the observation in Fig \ref{f:flows}, the increase in the number of coflows leads to more network congestion. Since SCASA adjusts the selected sources and flow ordering based on network congestion, it performs better than different benchmark solutions. The performance difference between SCASA and FLS increases from 19.8\% at 10 coflows to 24.5\% to 30 coflows. Similarly, the performance difference increases from 12.3\% to 16.7\% for CFLS, 17.9\% to 19.9\% for BAS, and 19.3\% to 33\% for FLORD.

\subsubsection{Changing average release time of flows}

SCASA has also been compared with other benchmark solutions by changing the average release time of flows, as shown in Fig \ref{f:release}. We have changed the scaling factor $T_{avg}$ of release time from 1 to 5. SCASA performs better than different benchmark solutions, even for different values of release time of flows. The increase in the average release time of flows impacts the value of the sum of CCT, thereby making the performance difference slightly larger. The performance difference between SCASA and FLS increases slightly from 21.2\% to 23.9\%. Similarly, there is a slight increase in performance difference for CFLS from 12.6\%  to 16.3\% and for BAS from 17\% to 24.4\%. The performance difference between SCASA and FLORD stays in the same range, around 25\% to 29\%, on changing the release time.

\subsubsection{Changing number of devices in network}

Fig \ref{f:devices} shows the performance comparison between SCASA and benchmark solutions on increasing the number of devices in the network from 30 to 150. The performance difference between SCASA and benchmark solutions becomes smaller when the network devices increase. It is expected as the number of coflows and flows will be kept constant while increasing the number of devices, leading to fewer network conflicts and a decreased value sum of CCT. The performance difference between SCASA and FLS decreases from 21.9\% at 30 devices to 6.3\% at 150 devices. Similarly, there is a decrease in performance difference from 14.6\% to 4.6\% for CFLS, 18.6\% to 6.7\% for BAS, and 31.9\% to 14.3\% for FLORD.

\subsubsection{Changing number of devices in network and coflows}

In previous comparisons, either the number of devices in the network or the number of coflows are kept constant. We have also compared SCASA with benchmark solutions by changing both the number of devices in the network and the number of coflows, as shown in Fig \ref{f:devcof}. The ratio of the number of devices in the network to the number of coflows is kept at $2:1$. The number of devices in the network is increased from 20 to 100 while changing the number of coflows from 10 to 50 correspondingly. It is observed that the performance difference between SCASA and benchmark solutions increases for FLS and CFLS and stays in the same range for BAS and FLORD. The performance difference between SCASA and FLS increases from 13.6\% to 22.2\%, while for CFLS, it increases from 3.7\% to 15.1\%. Since both FLS and CFLS do not consider network conflicts while deciding the source and flow order, the increase in network congestion leads to worse performance than SCASA. Even though BAS and FLORD consider network conflicts, they do not change the source selection, leading to inefficient flow ordering and worse performance than SCASA. The performance difference between SCASA and BAS stays around 13.9\% to 19.2\%, while for FLORD, it is 27\% to 36\%.

\subsection{Comparison with SCASA+FLORD}

We have also compared SCASA with SCASA+FLORD, a modified version of SCASA that also considers flow adjustment at individual links. SCASA+FLORD leverages the flow adjustment in FLORD after the source search and adjust heuristic in SCASA to improve performance. Our experiments show that we obtain a similar level of performance by SCASA+FLORD compared to SCASA with a much longer running time.  Fig \ref{f:scas+flord} shows the comparison conducted for the default parameter setting for 30 iterations. Overall, we can observe that for many iterations, both SCASA and SCASA-FLORD give the same result, which implies that flow adjustment did not lead to performance improvement in such cases. However, in some iterations, we observe that SCASA+FLORD leads to around 5\% or even more performance improvement in the value in the sum of CCT. Flow adjustment can provide some benefits depending on the parameter setting; however, it also significantly increases running time. 

 \begin{figure}[!h]
 \begin{subfigure}[t]{0.23\textwidth}
    \centering
\includegraphics[width=1\textwidth]{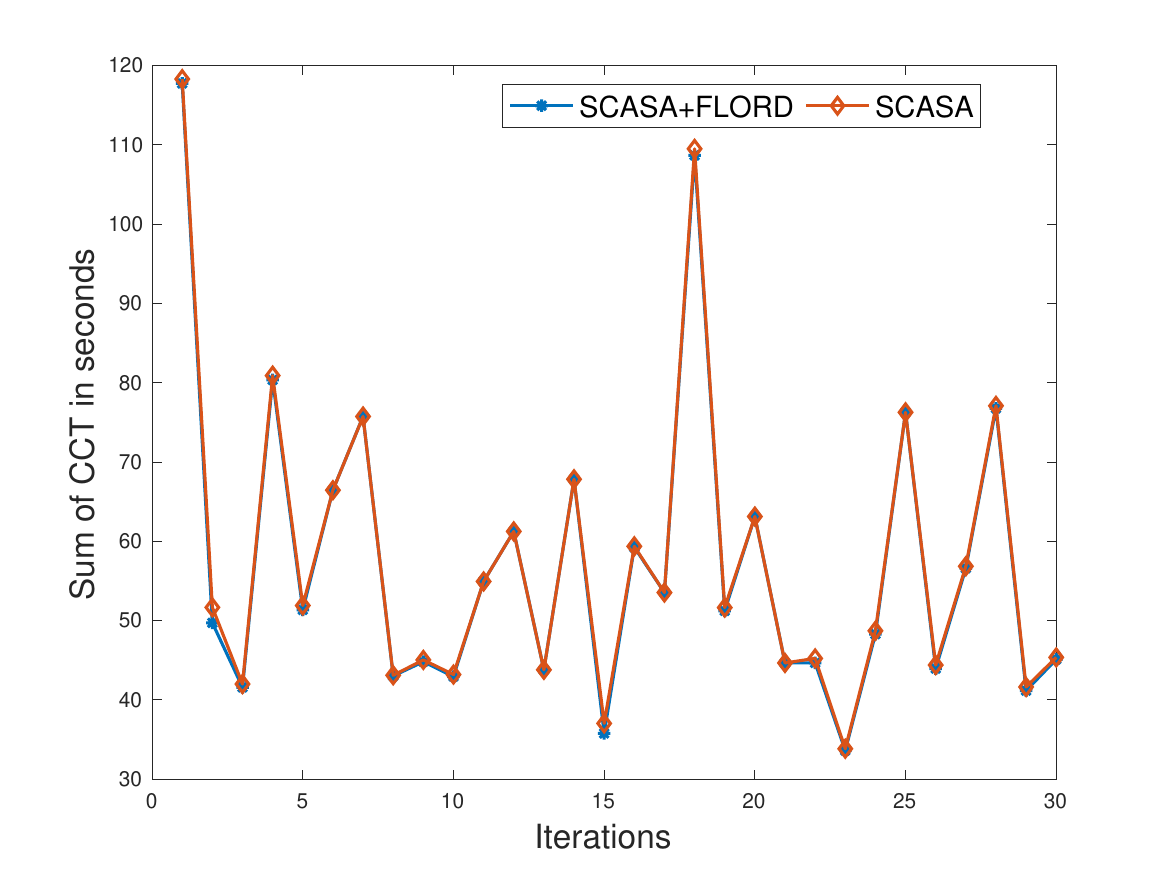}
\caption{Performance comparison with SCASA+FLORD for 40 devices in network}
\label{f:scas+flord}
\end{subfigure}  
~
 \begin{subfigure}[t]{0.23\textwidth}
    \centering
\includegraphics[width=1\textwidth]{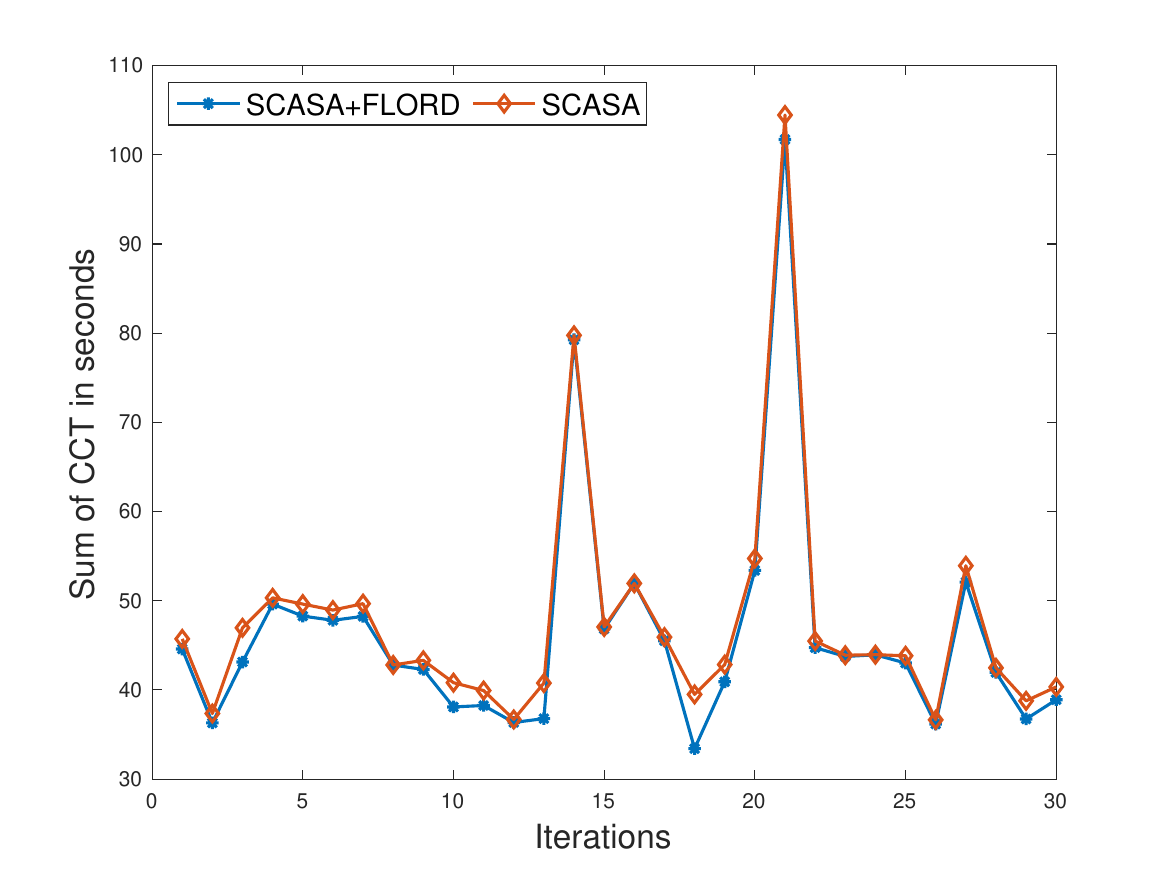}
\caption{Performance comparison with SCASA+FLORD for 120 devices in network}
\label{f:scasa+flord120}
\end{subfigure}  
 \caption{Performance comparison with benchmark solutions}
\label{f:devcof}
\end{figure} 

We have also done the comparison by increasing the number of devices in the network from 40 to 120, as shown in Fig \ref{f:scasa+flord120}. It is observed that the performance difference between SCASA-FLORD and SCASA increases when the number of devices in the network is increased. The performance difference is 8.89\% at the 90th percentile and can even go up to 18\% in some iterations. The increase in performance difference can be explained by the additional benefit brought by flow adjustment, especially in large networks with more links having dependent conflicts among flows that can be better resolved by considering link level adjustments instead of global flow ordering. In terms of running time, SCASA is significantly faster than SCASA+FLORD since no flow adjustment is required. The running time of SCASA is around 73\% less compared to SCASA+FLORD. 

\section{Related Works}

We classify the related works into three categories: coflow scheduling in data centers, network-aware task scheduling in edge computing, and data dissemination and offloading in edge computing. 

Coflow scheduling problem has been studied extensively in literature in the context of data centers \cite{chowdhury2014efficient} \cite{qiu2015minimizing}  \cite{chen2016optimizing} \cite{luo2016towards} \cite{zhou2019fast} \cite{ahmadi2020scheduling} \cite{chowdhury2019near} \cite{xu2020scheduling}. Most of these works assume the network model of a non-blocking switch connecting the devices, with network congestion only at ingress and egress ports. Varys \cite{chowdhury2014efficient} is one of the initial works to solve coflow scheduling in data centers by proposing the Smallest-Effective-Bottleneck-First (SEBF) heuristic to decide the order among coflows and Minimum-Allocation-for-Desired-Duration (MADD) for bandwidth allocation. Later, several works have proposed approximation algorithms for coflow scheduling \cite{qiu2015minimizing} \cite{zhang2019coflow} \cite{jahanjou2017asymptotically} \cite{shafiee2018improved} \cite{ahmadi2020scheduling} \cite{mao2021npscs}. The problem of scheduling dependent coflows has also been studied in \cite{shafiee2021scheduling} \cite{tian2019scheduling} \cite{tian2018scheduling}.  The work in \cite{wang2018multi} proposed multiple-attributes-based coflow Scheduling (MCS) mechanism for information-agnostic coflow scheduling. The work in \cite{xu2018optimizing} proposes Lever to optimize the cost-performance tradeoff for coflows in the geo-distributed data center.  

Few works have also considered scheduling coflows in general network topologies with bandwidth constraints \cite{li2016efficient}  \cite{chowdhury2019near} \cite{zeng2021scheduling}. The work in \cite{li2016efficient} proposed an efficient online solution for the multiple online coflow routing and scheduling problem. Another work in \cite{chen2018multi}  also studied multi-hop coflow routing and scheduling in the data center with leaf-spine topology. Related work has also studied the joint routing and bandwidth allocation problem for coflow scheduling \cite{shi2021coflow}. The work in \cite{chiang2021information} has also solved the grant assignment and transmission scheduling problem considering freshness information among coflows. Some recent works have also considered integrating coflow and circuit scheduling for optical networks \cite{wang2018integrating} \cite{zhang2020minimizing} \cite{tan2021regularization}. Coflow-aware job scheduling has also been studied in \cite{li2022co}, where a co-scheduler has been proposed for hybrid electrical/optical data center networks. Another work in \cite{huang2017exploiting} formulated a coflow placement problem and proposed a heuristic for selecting endpoints of coflows. Joint coflow placement and scheduling problems have also been studied in literature \cite{tan2019joint} \cite{zhao2020joint} \cite{li2020endpoint}.

Another category of work is network-aware task scheduling in edge computing, which considers the scheduling of underlying network flows (including routing, start time, and bandwidth allocated) to make task scheduling or computation offloading decisions. Some recent works such as \cite{funai2019computational}, \cite{al2016distributed}, \cite{hong2019qos}, \cite{hong2019multi}, etc. have addressed multi-hop computation offloading problem where they consider the routing path selection in multi-hop networks. Other works, such as \cite{munir2016network}, \cite{pu2015low}, \cite{rupprecht2017squirreljoin}, etc., have considered network bandwidth while making task scheduling decisions. The work in \cite{guo2019joint} considers the optimization of task placement and routing to minimize the coflow completion time. The work in \cite{wu2021joint} proposed a solution to jointly optimize the task placement, coflow bandwidth scheduling, and path choice for minimizing the average CCT in the intra-data center. Another work in \cite{chiang2019joint} studied joint cotask aware offloading and scheduling in mobile edge computing systems. Here, cotask is defined, similarly to coflows, where each job consists of a set of parallel tasks. In our previous works, we have also studied joint network flow scheduling and task scheduling problems for different application models, including multiple dependent tasks \cite{sahni2018data}, multiple independent tasks \cite{sahni2020multi}, and multiple directed acyclic graph (DAG) tasks \cite{sahni2020hop}. 

Several existing works have also studied data dissemination problems in edge computing to decide on different variables, including the placement, routing, and bandwidth allocated in distributing the data. The work in \cite{aral2018decentralized} proposes a decentralized method for the placement of replicas in edge computing. Another work in \cite{yuan2018toward} proposed solutions for content placement and sharing in vehicular networks. Several other works, such as \cite{liu2020fog}, \cite{yang2019efficient}, \cite{singh2020intent}, etc., have also studied data dissemination for vehicular networks. Data offloading in edge computing is another related problem that has been studied by many existing works in literature  \cite{dziyauddin2021computation} \cite{fantacci2020performance}. For example, work in \cite{mitsis2020data} proposed a game-theoretic approach for deciding the data offloading, i.e., how much data is to be offloaded and to which device. 

Compared to existing works, this paper considers coflow scheduling in a multi-hop network of edge devices by jointly optimizing the source selection of flows and flow ordering at network links. Existing works on coflow scheduling in data centers often assume non-blocking switches and focus on ingress and egress ports rather than congestion among flows at different links in the routing paths. Few existing works on coflow scheduling have considered general network topologies, however, these works do not jointly consider selecting the source of flows. Some works related to task scheduling have considered multi-hop paths in the network. However, they have focused more on the computation offloading aspects and ignored the congestion among flows at different links in the path. Finally, several works related to data dissemination focus on data placement, routing, and bandwidth allocation issues. These works do not consider coflow scheduling in a mesh network of edge devices connected using multi-hop paths.

 \section{Conclusion}
 
This paper studies the multi-source coflow scheduling problem in CEC. The problem jointly decides the source and order of flow to minimize the sum of CCT. First, we formulate the problem as a MINLP optimization problem, which is proven to be NP-hard. Then, we propose SCASA, an efficient heuristic to solve the problem that decides the source and order of flows by leveraging the knowledge of flows within and across coflows and network congestion at the links. SCASA was evaluated using numerical simulation and compared with several benchmark solutions. We have done comprehensive simulation experiments by varying different input parameters. Simulation experiments show that SCASA can achieve up to 83\% reduction in the value of the sum of CCT compared to different benchmark solutions that do not consider joint source and flow ordering solutions. One observation that can be made is that there is a tradeoff in terms of computation complexity and CCT performance. SCASA can achieve better performance than other benchmarks in reducing the value of the sum of CCT, however, with some compromise to computation time. 
 
This paper has proposed an offline solution that assumes global network knowledge. This work can be extended in the future by proposing an online solution assuming no a priori knowledge of coflows and data generated at the edge devices. One option is to propose a reinforcement learning-based solution where the action space can decide the source of flows and order of flows at different links. We can even extend this work to include bandwidth sharing at the links instead of limiting one flow at a link at one time by including the bandwidth used by each flow as another decision variable. Furthermore, a distributed solution can also be proposed to have improved scalability.

 \bibliographystyle{IEEEtran}
\bibliography{refer}

\vfill

\end{document}